\documentclass[11pt]{article}
\usepackage{xinyu,enumitem}
\newcommand{\free}{\operatorname{free}}

\title{A stochastic calculus approach to the oracle separation of $\BQP$ and $\PH$}
\author{Xinyu Wu\thanks{Computer Science Department, Carnegie Mellon University. \texttt{xinyuwu@cmu.edu}}
}
\date{July 5, 2020}

\begin{document}
\maketitle
\setlength\parskip{0.1em}

\begin{abstract}
After presentations of Raz and Tal's oracle separation of $\BQP$ and $\PH$ result, several people (e.g. Ryan O'Donnell, James Lee, Avishay Tal) suggested that the proof may be simplified by stochastic calculus. In this short note, we describe such a simplification.
\end{abstract}

\section{Introduction}\label{sec:intro}
A recent landmark result of Raz and Tal~\cite{RT19} shows there exists an oracle $A$ such that $\BQP^A \not\subseteq \PH^A$. Using a correspondence between $\PH$ and $\AC^0$ circuits, the question reduces to a lower bound against $\AC^0$ circuits. Concretely, it suffices to show that there exists a distribution $\calD$ over $\{-1,1\}^N$ such that
\begin{enumerate}[nosep]
\item For any $f:\{-1,1\}^N \to \{0,1\}$ computable by an $\AC^0$ circuit,
\[
    \abs*{\E[f(\calD)]-\E[f(\calU_N)]} \leq \frac{\polylog(N)}{\sqrt{N}},
\]
where $\calU_N$ is the uniform distribution on $N$ bits. The notation $\E[f(\calD)]$ means $\E_{\bx \sim \calD}[f(\bx)]$.

\item There exists a quantum algorithm $Q$ such that
\[
    \abs*{\E[Q(\calD)]-\E[Q(\calU_N)]} \geq \Omega\parens*{\frac{1}{\log{N}}}.
\]
\end{enumerate}
For details, we refer to Raz and Tal's paper~\cite{RT19}. $\calD$ in Raz and Tal's work is a truncated Gaussian. In this note, we will describe a construction of $\calD$ based on Brownian motion, which simplifies many details of the analysis.

\subsection{Stochastic calculus preliminaries}
We briefly review some stochastic calculus concepts used in the proof. See for instance~\cite[Chapter 7]{Oks03} for details.

\begin{definition}
An $N$-dimensional standard Brownian motion $\bB:[0,\infty) \x \R^N \to \R^N$ is a continuous-time stochastic process characterized by the following:
\begin{enumerate}[(i), nosep]
\item $\bB_0 = 0$ almost surely.
\item $\bB_{t+u} - \bB_t$ for $u \geq 0$ is independent of $\bB_s$ for $s < t$.
\item $\bB_{t+u} - \bB_t$ for $u \geq 0$ is distributed as an $N$-dimensional Gaussian with mean 0 and covariance matrix $u I_{N\x N}$.
\item $\bB_t$ is continuous almost surely.
\end{enumerate}
\end{definition}

We can describe a large class of stochastic processes, called \emph{\Ito/ diffusion} processes, by the solutions of stochastic differential equations of the following form:
\[
    d\,\bX_t = b(\bX_t)\,dt + \sigma(\bX_t)\,d\bB_t.
\]

\begin{definition}
Let $\bX$ be an \Ito/ diffusion. The \emph{infinitesimal generator} of $f$, is defined as
\[
    Af(x) = \lim_{t \to 0} \frac{\E_x[f(\bX_t)] - f(x)}{t}.
\]
We use the $\E_x[\cdot]$ notation to mean that we let $\bX_t$ evolve with starting point $x$.

If $f$ is twice continuously differentiable with compact support, we have the following expression for $Af$:
\[
    Af(x) = b(x)\cdot \nabla f(x) + \frac12 \tr(\sigma(x) \sigma^\top(x)\, \mathbf{\operatorname{H}}(x)),
\]
where $\mathbf{\operatorname{H}}$ is the Hessian of $f$.
\end{definition}
For example, the infinitesimal generator of a standard 1D Brownian motion is the Laplacian operator. For a Brownian motion with covariance matrix $\Sigma$, the infinitesimal generator would be $\tr(\Sigma \, \mathbf{\operatorname{H}}(x))$.

Next we state Dynkin's formula, which will be the main tool we use in the later proof.
\begin{theorem}[Dynkin's formula, {\cite[Theorem 7.4.1]{Oks03}}]\label{thm:dynkin}
Let $\bX$ be an \Ito/ diffusion, let $\tau$ be a stopping time with $\E[\tau] < \infty$, and let $f:\R^N \to \R^N$ be a twice continuously differentiable function with compact support. The following holds:
\[
    \E_x[f(\bX_\tau)] = f(x) + \E_x\bracks*{\int_0^\tau Af(\bX_s)\, ds},
\]

Moreover, if $\bX_\tau$ is bounded, Dynkin's formula with the same expression for $Af$ holds for $f$ which is twice continuously differentiable (without compact support).
\end{theorem}

\section{Reduction to a Fourier bound}
The main technical part of Raz and Tal's result~\cite{RT19} shows that, for a Boolean function $f:\{-1,1\}^N \to \{-1,1\}$ computable by an AC$^0$ circuit, and a multivariate Gaussian distribution $\bZ \in \R^N$,
\[
  |\E[f(\mathsf{trnc}(\bZ))]-\E[f(\bU_N)]| \leq O(\gamma\cdot \polylog(n)),
\]
where $\gamma$ is a bound on the (pairwise) covariance of the coordinates of $\bZ$, $\mathsf{trnc}$ truncates $\bZ$ so that the resulting random variable is within $[-1,1]^N$, and $\bU_N$ is the uniform distribution over $\{-1,1\}^N$. The important condition used here is that AC$^0$ has second level Fourier coefficients bounded by $\polylog(n)$, and that this holds under any restriction of the function.

Another natural way of viewing a multivariate Gaussian distribution is as the result of an \mbox{$N$-dimensional} Brownian motion stopped at a fixed time. We can also build the truncation into the stopping time. This allows us to use tools from stochastic calculus to analyze the distribution.

We first recall the definition of restrictions of Boolean functions.
\begin{definition}
Let $f:\{-1,1\}^N \to \R$ and let $\rho \in \{-1,1,*\}^N$. Let $\free(\rho)$ be the set of coordinates with $*$'s. We define the restriction of $f$ by $\rho$ as $f_\rho:\{-1,1\}^{N} \to \R$, and $f_\rho(x)$ is $f$ evaluated at $\rho$ with $x$ replacing the $*$'s in $\rho$.\footnote{Although $f_\rho$'s domain is $\{-1,1\}^N$, it only depends on the coordinates in $\free(\rho)$.}
\end{definition}

Henceforth, we also identify Boolean functions $f:\{-1,1\}^N \to \R$ with their multilinear polynomial representations (or Fourier expansions)
\begin{equation}\label{eqn:fourier-formula} \notag
    f(x) = \sum_{|S| \subseteq [N]} \hat f(S) \prod_{i \in S} x_i.
\end{equation}

We make some observations about Fourier coefficients. First, the Fourier coefficients of $f_\rho$ satisfy $\wh{f_\rho}(S) = 0$ for all $S \not\subseteq \free(\rho)$. We also have that
\begin{equation}\label{eqn:deriv-fourier}
    \hat{f}(S) = \partial_S f(0),
\end{equation}
where $\partial_S = \prod_{i \in S} \pt_i$ and $\pt_i = \frac{\partial}{\partial x_i}$ is the usual calculus derivative. Further, because $f$ is multilinear, for any $h \in \R \setminus\{0\}$ and any standard basis vector $e_i$ we have
\begin{equation}\label{eqn:deriv-difference}
    \pt_i f(x) = \frac{f(x + he_i) - f(x)}{h}.
\end{equation}

The following lemma is similar to~\cite[Claim A.5]{CHLT18}, which first appeared in~\cite{BB18} and~\cite[Claim 3.3]{CHHL19}.
\begin{lemma}\label{lem:restriction}
Let $f:\R^N \to \R$ be a multilinear polynomial. For any $x \in [-1/2,1/2]^N$, there exists a distribution $\calR_x$ over restrictions $\brho \in \{-1,1,*\}^N$, such that for any $i,j \in [N]$,
\[
     \pt_{ij}f(x) = 4\E_{\brho \sim \calR_x}\bracks*{\pt_{ij}f_{\brho}(0)}.
\]
\end{lemma}
\begin{proof}
We define $\calR_x$ as such: for each coordinate $i \in [N]$ we independently set
$\brho_i$ to be $1$ with probability $\frac14 + \frac{x_i}{2}$, to be $-1$ with probability $\frac14 - \frac{x_i}{2}$, and to be $*$ with probability $\frac12$.

Using that $f$ is a multilinear polynomial, and that the coordinates are independent, we deduce that for any $y \in \R^N$, $f(x+y) = \E_{\brho \sim \calR_x}\bracks*{f_\brho(2y)}$. Then, using~\Cref{eqn:deriv-difference},
\begin{align*}
    \pt_{ij}f(x) &= f(x+e_i + e_j) - f(x+e_i) - f(x+e_j)+f(x)\\
    &= \E_{\brho\sim\calR_x}\bracks*{f_\brho(2e_i + 2e_j) - f_\brho(2e_j) - f_\brho(2e_i) +f_\brho(0)}
    = 4\E_{\brho\sim\calR_x}\bracks*{\pt_{ij}f_{\brho}(0)}. \qedhere
\end{align*}
\end{proof}
We now show the main result, which is a restatement of~\cite[Therorem A.7]{CHLT18} and~\cite[Theorem 2.4]{RT19}.
\begin{theorem}\label{thm:main}
Let $f:\{-1,1\}^N \to \{-1,1\}$ be a Boolean function, and let $t > 0$ such that for any restriction~$\rho$,
\[
    \sum_{\substack{S \se [N] \\ |S| = 2}} |\wh{f_\rho}(S)| \leq t.
\]
Let $\gamma > 0$ and let $\bX$ be an $N$-dimensional Brownian motion with mean 0 and covariance matrix $\Sigma$, in the sense that $\E[(\bX_t)_i] = 0$ for all $i \in [N]$, and $\Cov((\bX_t - \bX_s)_i,(\bX_t - \bX_s)_j) = (t-s) \Sigma_{ij}$. Further assume that $|\Sigma_{ij}| \leq \gamma$ for $i \ne j$.

Let $\ep > 0$ and define the stopping time
\[
    \tau \coloneqq \min\,\{\ep, \text{ first time that $\bX_t$ exits } [-1/2,1/2]^N\}.
\]
Then, identifying $f$ with its multilinear expansion, we have
\begin{equation*}
    \abs*{\E[f(\bX_\tau)] - \E[f(\bU_n)]}\leq 2\ep \gamma t.
\end{equation*}
\end{theorem}
\begin{proof}
First, we note that $\E[f(\bU_N)] = f(0)$. Next, let $\sigma = \Sigma^{1/2}$. 
$\bX$ satisfies the stochastic differential equation
\[
    d\bX_t = \sigma d \bB_t.
\]
Note that $\bX_\tau$ is always within $[-1/2,1/2]^N$. We can apply~\Cref{thm:dynkin}
\[
    \E[f(\bX_\tau)] - f(0) = \E\bracks*{\int_0^\tau \frac12 \sum_{i,j\in [N]} \Sigma_{ij} \pt_{ij}f(\bX_s)\,ds}.
\]
Then, we upper bound $\tau \leq \ep$, and use that $\pt_{ii} f = 0$ for all $i \in [N]$ because $f$ is multilinear, to get
\begin{align*}
    |\E[f(\bX_\tau)] - f(0)| &\leq \ep \E\bracks*{\sup_{s \in [0,\tau]}\abs*{\frac12 \sum_{i,j\in[N]} \Sigma_{ij}\pt_{ij}f(\bX_s)}}\\
    &\leq \frac{\ep\gamma}2 \sup_{x \in [-1/2,1/2]^N}\sum_{i \ne j}\abs*{\pt_{ij}f(x)}
    \\
    &= 2\ep\gamma \sup_{x \in [-1/2,1/2]^N}\sum_{i\ne j}\abs*{ \E_{\brho \sim \calR_x}\bracks*{\pt_{ij}f_{\brho}(0)}} &&\text{(\Cref{lem:restriction})}\\
    &\leq 2\ep\gamma \sup_{x \in [-1/2,1/2]^N}\E_{\brho\sim\calR_x}\bracks*{\sum_{i\ne j}\abs*{\pt_{ij}f_{\brho}(0)}}\\
    &\leq 2\ep\gamma \sup_{x \in [-1/2,1/2]^N}\E_{\brho\sim\calR_x}\bracks*{\sum_{\substack{S \se \free(\brho)\\|S| = 2}}\abs*{\hat f_{\brho}(S)}} && \text{(\Cref{eqn:deriv-fourier})}\\
    &\leq 2\ep\gamma t. \qedhere
\end{align*}
\end{proof}

\section{Application to the oracle separation of \texorpdfstring{$\BQP$}{BQP} and \texorpdfstring{$\PH$}{PH}}
We now use~\Cref{thm:main} to construct $\calD$ as described in~\Cref{sec:intro}.

\paragraph*{The distribution $\calD$.}
Let $N = 2n$, where $n$ is a power of $2$, and
\[
    \Sigma \coloneqq \begin{pmatrix}
    I_n & H_n\\ H_n & I_n
    \end{pmatrix},
\]
where $H_n$ is the Walsh--Hadamard matrix.
Now we define $\bX$ and $\tau$ as in~\Cref{thm:main}, with $\ep = 1/(8\ln N)$, and our distribution $\calD$ will be the distribution defined by $\bX_\tau$.
At each time $t$, we can also look at $\bX_t$ as a pair of random variables in $\R^n$, $(\bx_t,\by_t)$ such that $\by_t$ is the Hadamard transform of $\bx_t$.

\paragraph*{$\AC^0$ lower bound.}
Tal showed that~\cite[Theorem 37]{Tal17} there exists a universal constant $c$ such that every function $f:\{-1,1\}^N \to \{-1,1\}$ computable by an $\AC^0$ circuit with at most $(\ln N)^\ell$ gates and depth $d$ satisfies
\begin{equation}\label{eqn:ac0-fourier} \notag
    \sum_{\substack{S \subseteq [N]\\|S|=k}} |\hat f(S)| \leq (c\cdot \ln^\ell N)^{(d-1)k}.
\end{equation}
Since $\AC^0$ is closed under restrictions, we can apply~\Cref{thm:main} with $\ep = 1/(8\ln N)$ and $\gamma = \frac{1}{\sqrt n}$, to deduce that
\[
    |\E[f(\bX_\tau)] - f(0)| \leq \frac{\polylog N}{\sqrt N}.
\]
\paragraph*{Quantum algorithm.}
Finally, we show that a quantum algorithm can distinguish $\calD$ from the uniform distribution.
This is virtually identical to the argument in~\cite[Section 6]{RT19}, but we can again use some stochastic calculus tools on the stopping time built into the distribution. Using the Forrelation query algorithm, there is a quantum algorithm $Q$ with inputs $x,y \in \{- 1,1\}^n$ which accepts with probability $(1+\phi(x,y))/2$, where
\[
     \phi(x,y) \coloneqq \frac1n \sum_{i,j \in [n]} x_i \cdot H_{ij} \cdot y_j.
\]
We show the following proposition~\cite[Claim 6.3]{RT19}, which implies the existence of a $O(\log N)$-time quantum algorithm distinguishing $\calD$ from uniform with one query.
The quantum algorithm is described in more detail in~\cite[Section 3.2]{Aar10}.
\begin{proposition}
$\E_{(\bx,\by) \sim \calD}[\phi(\bx,\by)] \geq \frac\ep4.$
\end{proposition}
\begin{proof}
By the linearity of expectation and optional sampling theorem,
\begin{align*}
\E_{(\bx,\by) \sim \calD}[\phi(\bx,\by)] &= \frac1n \sum_{i,j \in [n]} H_{ij}\cdot\E[\bx_i\cdot \by_j]\\
&= \frac1n \sum_{i,j \in [n]} H_{ij}\cdot \E[\tau]\cdot H_{ij} = \E[\tau].
\end{align*}
By Markov's inequality,
\[
    \E[\tau] \geq \frac{\ep}{2} \Pr[\tau > \tfrac\ep2].
\]
If $\tau \leq \frac\eps2$, it must be the case that the path exits $[-1/2,1/2]^N$ no later than $\frac\eps2$. Hence, we can upper bound
\[
    \Pr\bracks*{\tau \leq \tfrac\ep2} \leq N \cdot \Pr\bracks*{\text{1st coordinate of $X_t$ exits $\bracks*{-\tfrac12,\tfrac12}$ earlier than } \tfrac\ep2}.
\]
Each coordinate of $\bX$ is a standard 1D Brownian motion since $\Sigma_{ii} =1$ for all $i$. An application of Doob's martingale inequality (e.g.~\cite[Proposition II.1.8]{RY99}) tells us that, for a standard 1D Brownian motion $\bB_t$,
\[
    \Pr\bracks*{\sup_{0 \leq t \leq \ep/2} |\bB_t| \geq \frac12} \leq 2e^{-1/4\ep} = 2 e^{-2\ln N} \leq \frac{1}{2N} \quad \text{for } N \geq 4.
\]
Therefore, $\Pr[\tau \leq \frac \ep 2] \leq \frac12$, so $\E[\tau] \geq \frac\ep 4$.
\end{proof}

\section{Acknowledgments}
I would like to thank Ryan O'Donnell and Avishay Tal for helpful discussions and their suggestions concerning an early draft. Thanks also to Gregory Rosenthal and anonymous reviewers for helpful comments.

\bibliographystyle{alpha}

\end{document}